\documentclass[11pt,twoside]{article}
\usepackage[letterpaper,margin=1in]{geometry}
\usepackage[utf8x]{inputenc}
\usepackage{amssymb}
\usepackage{graphicx}
\usepackage{verbatim}
\title{Generic Single Edge Fault Tolerant Exact Distance Oracle}
\author{Manoj Gupta\\
   IIT Gandhinagar, India\\
   {\small\texttt{gmanoj@iitgn.ac.in}}
 \and
Aditi Singh\\
   IIT Gandhinagar, India\\
{\small\texttt{aditi.singh@iitgn.ac.in}}
}

\usepackage{amsmath,amssymb,amsthm,xpatch}

\usepackage{dsfont}
\usepackage{xargs}
\usepackage{todonotes}
\usepackage{tikz}
\usepackage[boxruled,lined,linesnumbered]{algorithm2e}
\usepackage{amssymb}
\usepackage{wrapfig}
\usepackage{mathrsfs}
\usepackage{enumerate}
\usepackage{caption}
\usepackage{microtype}
\usepackage{soul}
\usepackage{enumitem}

\newtheorem{theorem}{Theorem}
\newtheorem{lemma}[theorem]{Lemma}

\newtheorem{definition}[theorem]{Definition}
\newtheorem{corollary}[theorem]{Corollary}

\newif\iflong
\longtrue

\newcommand{\TT}{\mathcal{T}}

\newcommand{\conc}{\diamond}
\newcommand{\DET}{\textsc{Detour}}
\newcommand{\UNQ}{\textsc{Unique}}
\newcommand{\FF}{\textsc{First}}
\newcommand{\LL}{\textsc{Last}}
\newcommand{\SBFS}{$\sigma$-\textsc{BFS}}
\newcommand{\BFS}{\textsc{BFS}}
\newcommand{\TZE}{\mathcal{R}_1}
\newcommand{\TON}{\mathcal{R}_1}
\newcommand{\TTW}{\mathcal{R}_2}
\newcommand{\RR}{\mathcal{R}}
\newcommand{\BP}{\mathcal{B}}

\newcommand{\INT}{\textsc{Int}}
\newcommand{\BST}{\textsc{Bst}}
\newcommand{\RMQ}{\textsc{Rmq}}
\newcommand{\HE}{\textsc{Heavy}}
\newcommand{\LI}{\textsc{Light}}
\newcommand{\QQ}{\textsc{Q}}

\begin{document}

\maketitle

\begin{abstract}
 Given an undirected  unweighted graph $G$ and a source set $S$ of $|S| = \sigma
$ sources, we want to build a data structure which can process
the following query {\sc Q}$(s,t,e):$ find the shortest distance from $s$ to $t$ avoiding an edge $e$,  where $s \in S$ and
$t \in V$. When $\sigma=n$, Demetrescu, Thorup, Chowdhury and Ramachandran (SIAM Journal of Computing, 2008) designed an algorithm
with $\tilde O(n^2)$ space\footnote{$\tilde
O(\cdot)$ hides poly $\log n$ factor.} and $O(1)$ query time. A natural open question is to
generalize this result to any number of sources. Recently,
Bil{\`o} et. al. (STACS 2018) designed a data-structure of size $\tilde O(\sigma^{1/2}n^{3/2})$ with the query time of $O(\sqrt{n\sigma})$ for the above problem.  We improve their result by designing
a data-structure of size $\tilde O(\sigma^{1/2} n^{3/2})$ that can answer queries in $\tilde O(1)$ time.

In a related problem of finding fault tolerant subgraph, Parter and Peleg (ESA 2013) showed that
if  detours of  {\em replacement} paths ending at a vertex $t$  are disjoint,
then the number of such paths is $O(\sqrt{n\sigma})$.
This eventually gives a bound of $O( n \sqrt{n \sigma}) = O(\sigma^{1/2}n^{3/2})$
for their problem. {\em Disjointness of detours} is a very crucial property used
in the above result. We show a similar result for a subset of replacement path which
\textbf{may not} be disjoint. This result is the crux of our paper and may be
of independent interest.

\end{abstract}

\newpage
\section{Introduction}
Real life graph networks like communication network or road
network are prone to link or node failure. Thus, algorithms
developed for these networks must be resilient to failure.
For example, the shortest path between two nodes may change
drastically even if a single link fails. So, if the problem
forces us to find shortest paths in the graph, then it should
find the next best shortest path after a link failure. There
are many ways to model this process: one of them is {\em
fault-tolerant graph algorithm}. In this model,  we have to preprocess a graph $G$  and
build a data-structure that can compute a property of the
graph after any $k$ edges/vertices of the graph have failed.
 Note the difference between this model and
{\em dynamic graph model}. In a dynamic graph algorithm,
 we have to {\em maintain} a property of a continuously
changing graph. However,
in the fault tolerant model, we expect the failure to be
repaired readily and
restore our original graph.

In this paper, we study the shortest path problem in the
fault tolerant model. Formally, we are given an undirected
and unweighted graph $G$ and a source
set $S$ of $|S| = \sigma
$ sources. We want to build a data structure which can process
the following query {\sc Q}$(s,t,e):$ find the shortest
distance from $s$ to $t$ avoiding an edge $e$,  where $s
\in S$ and
$t \in V$. Such a data-structure is also called a {\em distance
oracle}.
When there are $n$ sources, Demetrescu et al. \cite{DemetrescuTCR08}
designed an oracle that can find the shortest path between
any two vertices
in $G$ after a single vertex/edge failure in $\tilde O(n^2)$
space and $O(1)$
query time. Recently,
Bil{\`o} et. al. \cite{BiloCGLP17} generalized this result
to any number of sources. They designed a data-structure
of size $\tilde
O(\sigma^{1/2}n^{3/2})$ with the query time of $O(\sqrt{n\sigma})$
for the above problem.

To understand our problem, we should also understand a closely
related problem of finding
{\em fault tolerant subgraph}.
Here, we have to find a subgraph of $G$ such that BFS
tree from $s \in S$ is preserved in the subgraph after any
edge
deletion.  In an unweighted graph, a BFS tree preserves
the shortest path from $s$ to all vertices in $G$.
Parter and Peleg \cite{ParterP13} showed that a subgraph
of
size $O(\sigma^{1/2} n^{3/2})$ is both necessary and sufficient
to
solve the above problem. The above result indicates that
there should
be a better fault-tolerant distance oracle for any value
of $\sigma$.

Inspired by this result,
we generalize the result of \cite{DemetrescuTCR08} to any
number of sources --
by showing that there exists a distance oracle of size $\tilde
O(\sigma^{1/2} n^{3/2})$
which can answer queries in $\tilde O(1)$ time. Note that
our result nearly matches
the space bound achieved by Parter and Peleg\cite{ParterP13}
-- up to poly$\log n$ factors.
We now state the main result of this paper formally:
\begin{theorem}
\label{thm:maintheorem}
  There exists a data-structure of size $\tilde O(\sigma^{1/2}n^{3/2})$
  for multiple source single fault tolerant exact distance
oracle
  that can answer each query in $\tilde O(1)$ time.
\end{theorem}

This generalization turns out to be much more complex than
the result in \cite{DemetrescuTCR08}.
Indeed, the techniques used by Demetrescu et al. \cite{DemetrescuTCR08}
are also used
by us to weed out {\em easy replacement} paths. To take
care of other paths,
we take an approach similar to Parter and Peleg\cite{ParterP13}.
They used the following trick:
if the {\em detour} of replacement paths are {\em disjoint},
then the number of such
paths can be bounded easily by a {\em counting argument}.
The main challenge is
then to show that paths in question are indeed disjoint
-- which is also easy
in their problem. We use a technique similar to above --
however, our paths are not disjoint,
they {may intersect}.  We believe that this technique can
be of independent interest and may be used in solving
closely related fault tolerant subgraph problems.

\subsection{Related Work}
Prior to our work, the work related to fault tolerant
distance oracle was limited to two special cases, $\sigma
=1$ or $\sigma = n$.
As stated previously, Demetrescu et al. \cite{DemetrescuTCR08}
designed a single fault
tolerant distance oracle of size $\tilde O(n^2)$ with a
query time of $O(1)$.
The time to build the data-structure is $O(mn^2)$ --  which
was improved to
$O(m n \log n)$ by Bernstein and Karger \cite{BernsteinK09}.
The above result also works for a directed weighted graph.
Pettie and Duan \cite{DuanP09} were able to extend this
result to two vertex faults.
The size and query time of their distance oracle is $\tilde
O(n^2)$ and $\tilde O(1)$
respectively.
If the graph is weighted, then Demetrescu et al. \cite{DemetrescuTCR08}
showed that
there exists a graph in which a single vertex fault tolerant
distance oracle will
take $\Omega(m)$ space.
Recently, Bil{\`o} et. al. \cite{BiloCGLP17} designed the
following data-structure:  for every $S,T \subseteq V$,
a data-structure of size $\tilde O(n \sqrt{|S||T|})$ and
query time $O(\sqrt{|S||T|})$, where the query asks for
the shortest distance from $s \in S$ to $t \in T$ avoiding
any edge. If $|S| = \sigma$ and $|T| = n$, then the size
of their data-structure is $\tilde O(\sigma^{1/2}n^{3/2})$
and the query time is $O(\sqrt{n\sigma})$.

The next set of results are not {\em exact} but {\em approximate},
that is, they return
an approximate distance (by a multiplicative {\em stretch}
factor)
between two vertices after an edge/vertex fault. Also, these
oracles work  for a  single source only. Baswana and Khanna
\cite{KhannaB10} showed
that a 3-stretch  single source single fault tolerant distance
oracle of size $\tilde O(n)$ can be built in
$\tilde O(m + n)$ time and a constant query time.
Bil{\`{o}} et. al.\cite{BiloGLP16} improved the above result:
a distance oracle with stretch 2 of size
$O(n)$ and $O(1)$ query time.
In another result, Bil{\`{o}} et. al. \cite{BiloG0P16}
designed a $k$ fault
tolerant distance oracle of size $ \tilde O(kn)$ with a
stretch factor of $(2k+1)$
that can answer queries in $\tilde O(k^2)$ time. The time
required to construct this
data-structure is $ O(kn\alpha(m, n))$, where $\alpha(m,n)$
is the inverse of the Ackermann's function.
If the graph is unweighted,
then Baswana and Khanna\cite{KhannaB10}  showed that a $(1+\epsilon)$-stretch
single source
fault tolerant distance oracle of size $\tilde O(\frac{n}{\epsilon^3})$
can be built in $O(m \sqrt{n/\epsilon})$ time
and  a constant query time.  Bil{\`{o}} et. al \cite{BiloGLP16}
extended this result for
weighted graph by designing a distance oracle with stretch
$(1+\epsilon)$ of size
$O(\frac{n}{\epsilon} \log \frac{1}{\epsilon} )$ and a logarithmic
query time.

There is another line of work, called the {\em replacement
path} problem. In this problem,   we are given a source
$s$ and destination $t$ and for each edge $e$ on the shortest
 $st$ path, we need to find shortest $s$ to $t$ path
avoiding $e$. The problem can be generalized to finding
$k$ shortest $s$ to $t$ path avoiding $e$. The main goal
of this problem is to find all  shortest paths as fast
as possible.
Malik et al. \cite{MalikMG89}  showed that in an undirected
graphs, replacement paths can be computed in
 $O(m + n \log n)$ time. For
 directed, unweighted graphs, Roditty and Zwick \cite{RoddityZ12}
designed an algorithm that  finds all
replacement paths in $O(m
\sqrt n)$ time.
For the $k$-shortest paths problem, Roditty \cite{Roditty07}
presented
an algorithm with an approximation ratio
3/2, and the running time  $O(k(m\sqrt n+n^{3/2}
\log n))$.
 Bernstein \cite{Bernstein10} improved the above result
to get an approximation factor of $(1+\epsilon)$ and running
time $O(km/\epsilon)$.  The same paper also gives an improved
algorithm for the approximate $st$ replacement path algorithm.
See also \cite{GrandoniW12,Williams11,WeimannY10}.

As mentioned previously, a problem closely related to our problem is the fault tolerant subgraph problem. The aim of this problem is to
find a subgraph of $G$ such that BFS
tree from $s \in S$ is preserved in the subgraph after any
edge
deletion.
Parter and Peleg~\cite{ParterP13} designed an algorithm
to compute
single fault tolerant BFS tree with $O(n^{3/2})$ space.
They also showed their result can be easily extended
to multiple source with $O(\sigma^{1/2}n^{3/2})$ space.
Moreover, their upper bounds were complemented by matching
lower bounds for both their results.
This result was later extended to dual fault BFS tree  by
Parter~\cite{Parter15} with $O(n^{5/3})$ space.
Gupta and Khan \cite{GuptaK17} extended the above result
to
multiple sources with $O(\sigma^{1/3} n^{5 /3})$ space.
All the above results are optimal due to a result
by Parter~\cite{Parter15}  which states that a multiple
source $k$
fault tolerant BFS structure requires  $\Omega(\sigma^{\frac{1}{k+1}}n^{2-\frac{1}{k+1}})$
space.
Very recently, Bodwin et. al. \cite{BodwinGPW17} showed
the existence of a $k$ fault tolerant
BFS structure of size $\tilde O(k\sigma^{1/2^k}n^{2-1/2^k})$.

Other related problems include fault-tolerant DFS and fault
tolerant reachability.
Baswana et al.~\cite{BaswanaCCK16} designed an $\tilde{O}(m)$
sized fault tolerant data structure
that reports the DFS tree of an undirected graph after $k$
faults in $\tilde{O}(nk)$ time.
For single source reachability, Baswana et al. \cite{BaswanaCR16}
designed an algorithm
that finds a  fault tolerant reachability subgraph for
$k$ faults using $O(2^k n)$ edges.

\iflong

\subsection{Comparison with Previous Technique}
Our technique should be directly compared to the technique
in Bil{\`{o}} et. al.\cite{BiloCGLP17}.
We discuss their work when $|S|=1$ and $|T|=n$, that is
we want a single source
fault tolerant distance oracle. Consider any $st$ path where
$t \in V$.
For the last $\sqrt n \log n$ edges in this $st$ path (edges
from the vertex $t$),
we explicitly store the shortest replacement paths avoiding
these edges.
For the remaining replacement paths, note that the length
of these replacement
paths is always $\ge \sqrt n \log n$. So, we sample a set
$R$ of $O(\sqrt n)$ vertices.
With very high probability, one of our sampled vertices will
lie on the replacement
path (Lemma 1, \cite{BiloCGLP17}).
We can also show that the shortest path from the sampled vertex
to $t$
will never contain the edge avoided by these remaining replacement
paths (Lemma 3, \cite{BiloCGLP17}).
We can store all  shortest paths from $v \in R$ to $t
\in V$ using space $O(n^{3/2})$.

Thus, we have reduced the problem of finding a replacement
paths from $s$ to vertices in $V$,
to finding replacement path from $s$ to vertices in $R$.
This reduction is useful as
there are just $ O(\sqrt n)$ vertices in $R$. Fortunately,
there already exists a
data-structure \cite{DemetrescuTCR08}, say $D$, that can
solve the reduced
problem in $O(n^{3/2})$ space and $O(1)$ query time.

Now the query algorithm in  \cite{BiloCGLP17} is straightforward.
Consider the query $Q(s,t,e)$.
If $e$ is one of  the last $\sqrt n \log n$ edges on $st$
path, then we have already stored
the replacement path.
Else, we know that the replacement path avoiding $e$ must
pass through a vertex in $R$. Unfortunately,
we don't know that vertex. So, we try out all the vertices
in $R$. That is, for each $v \in R$,
we find the shortest path from $s$ to $v$ avoiding $e$ (using
data-structure $D$ in $O(1)$ time)
and add it to shortest $vt$ distance. Thus, we have to return
the minimum of all the computed
shortest paths. This gives us the running time of $|R| =
 O(\sqrt n)$.

 To improve upon the techniques in \cite{BiloCGLP17}, we
use the following
strategy: we also sample $R$ vertices of size $\tilde O(\sqrt
n)$. Instead of looking at all
the vertices in $R$, we concentrate on the vertex of $R$
that lies on the $st$ path, say $v$.
If the replacement path passes through $v$, then we can
 find it in $O(1)$ time using $D$ (as done in \cite{BiloCGLP17}).
The main novel idea of this paper is to show that the number
of replacement paths that do not
pass through $v$ is $O(\sqrt n)$. This helps us in reducing
the running time from $O(\sqrt n)$
to $\tilde O(1)$. Moreover, we show that this technique
can be generalized to any number of sources.
\fi


\section{Preliminaries}

We use the following notation throughout the paper:

\begin{itemize}[noitemsep,nolistsep]
\item $xy$ :\ Given
two vertices $x$ and $y$, let $xy$ denote a path between
$x$ and $y$. Normally this path will be the shortest path
from $x$ to $y$ in $G$. However, in some places in the paper, the
use of $xy$ will be clear from the context.

\item $|xy|$ :\ It denotes the number of edges in the path
$xy$.

\item $(\cdot \conc \cdot)$ :  Given two paths
$sx$ and $xt$, $sx \conc xt$ denotes the concatenation of
paths $sx$ and $xt$.

\item {\em after or below/before or above x} : We will assume that the $st$ path (for $s\in S$ and $t \in V$)  is drawn from top to bottom. Assume that $x \in st$. The term {\em after or below $x$} on $st$ path refers to the path $xt$. Similarly {\em before or above x} on $st$ path refers to the path $sx$.

\item {\em replacement path}: The shortest path that avoids any given edge is called a
{\em replacement} path.
\end{itemize}

\section{Our Approach}

\noindent  We will randomly select a set of terminals $\TT$ by sampling
each vertex with probability $\sqrt\frac{\sigma}{n}$. Note
that the size of $\TT$ is $\tilde O(\sqrt {\sigma n})$ with high probability.
For a source $s$ and $t \in V$, let $t_s$ be the last terminal encountered on the $st$ path.
The following lemma is immediate:
\begin{lemma}
\label{lem:lower}
If $|st| \ge c \sqrt{\frac{n}{\sigma}} \log n$ $(c \ge 3)$,
then $|t_st| = \tilde O(\sqrt{\frac{n}{\sigma}})$ with a
very high probability for all $s \in S$ and $t \in V$.
\end{lemma}

\iflong
  \begin{proof}
  Let $E_{s,t}$ be the event  that none of the  last $c \sqrt{\frac{n}{\sigma}} \log n$
  vertices on $st$ path are in $\TT$.  So,
  $\mathsf{P}[E_{st}\ \text{occors} ] = (1 - \sqrt{\frac{\sigma}{n}})^{c \sqrt{\frac{n}{\sigma}}
  \log n} \le \frac{1}{n^c}.$ Using union bound, $\mathsf{P}[ \cup_{s,t} E_{s,t}  ~\text{occors}] \le n \sigma \frac{1}{n^c} \le \frac{1}{n^{c-2}}$.
  Thus, with a very high probability  $|t_st| = \tilde O(\sqrt{\frac{n}{\sigma}})$  for all $s \in S$ and $t \in V$.

  \end{proof}
\fi

Let $G_p$ denote the graph where each edge is perturbed by a weight function that ensures unique shortest paths. Our
$st$ path is the shortest $s$ to $t$ path in $G_p$, let
us denote its  length by $|st|_p$. Note that $G_p$ contains a unique shortest path between any two vertices,
even the ones that avoid an edge -- such a graph has been used before in related problems
\cite{BernsteinK09,ParterP13,HershbergerS01}. We can use $G_p$ even to find all the
replacement paths. However, we want our replacement
paths to have  other nice property, that is,
{\em the length replacement paths(without perturbation) from $s$ to $t$ are different}.
This property is not satisfied by replacement paths in $G_p$. We employ another
simple strategy to find a replacement path. Following \cite{GuptaK17}, we
define preferred replacement paths:

\begin{definition}
A path $P$ is called a \textbf{preferred} replacement path from $s$ to $t$ avoiding $e$ if
(1)  it diverges and merges the $st$ path {\em just once}
(2)  it divergence point from the $st$ path is as close to $s$ as possible
(3)  it is the shortest path in $G_p$ satisfying (1) and (2).
\end{definition}

\noindent The replacement path has to diverge from the $st$ path before $e$. Ideally,
we want a replacement path that diverges from $st$ path as close to $s$ as possible.
This is a crucial feature which will ensure that all  replacement paths from
$s$ to  $t$ have  different lengths. The first condition ensures that we do
not diverge from $st$ path just to get a higher point of divergence. If many shortest paths are diverging from a same vertex, the third condition
is used to break ties. In the ensuing discussion, we will assume that  we are always working with a preferred replacement path.

The initial $st$ path is found out by finding the unique shortest path in $G_p$.
Consider the query
${\sc Q}(s,t,e)$. If the failed edge $e$ does not lie on $st$ path, then we can report $|st|$ as the shortest distance from $s$ to $t$ avoiding $e$. To this end, we should be able to check whether $e$ lies in the shortest path from $s$ to $t$. At this point, we will use the property of graph $G_p$. If $e(u,v)$ lies in $st$ path, then we have to check if $u$ and $v$ lie on $st$ path. To this end, we check if $|su|_p + |ut|_p = |st|_p$ and $|sv|_p + |vt|_p = |st|_p$.
If both the above two equations are satisfied then the $st$ path passes through $e$ (as the shortest path from $u$ to $v$ is 1). We can also find whether $u$ or $v$ is closer to  $s$ on $st$ path. Without loss of generality assume that $u$ is closer to $s$ than $v$ on $st$ path.

However, we do not have space to store all these distances. Specifically, the second term on the LHS of  above two equations mandates that we store the distance of every pair of vertices in the graph. This implies that the size of our data structure is
$O(n^2)$ which is not desirable.

To solve the above problem, we observe that if $e$ lies in the  $t_st$ path, then we have just enough space to store this fact. So, given any $e$, we can easily find if $e \in t_st$. If $e \in st_s$, then we know that $|su|_p+|ut_s|_p+|t_st|_p = |st|_p$ and $|sv|_p+|vt_s|_p+|t_st|_p
= |st|_p$. This equality is easier to check with the space at hand. So, we have the following two cases:

\begin{enumerate}[noitemsep,nolistsep]
\item(Near Case) $e$ lies on $t_st$.
\item(Far Case) $e$ lies on $st_s$.
\end{enumerate}

\subsection{Handling the Near Case}

For each $e(u,v) \in t_st$, let $P_e$ be the preferred replacement path from $s$ to $t$ avoiding $e$. We put $(e, |P_e|)$ in a balanced binary search tree $\BST(s,t)$ with the key being $e$.  Given any query $\sc{Q}(s,t,e)$, we now need to check if $e$ lies in $\BST(s,t)$. This can be done in $\tilde O(1)$ time and the length of the preferred replacement path  can be reported.

The space required for $\BST(s,t)$ is directly proportional to the size of path $t_st$. By Lemma \ref{lem:lower}, we know that $|t_st| = \tilde O(\sqrt{\frac{n}{\sigma}})$. Thus, the size of $\BST(s,t)\ = \tilde O(\sqrt{\frac{n}{\sigma}})$. This implies that the cumulative size of all the associated binary search tree  is $\cup_{t \in V}\cup_{s \in S} |t_st| = \tilde O( n \sigma \sqrt{\frac{n}{\sigma}}) = \tilde O(\sigma^{1/2} n^{3/2})$.
\subsection{Handling the Far Case}

We first need to check if $e \in st_s$. To this end we use the following data-structures.
\begin{itemize}[leftmargin=*,noitemsep,nolistsep]
\item $B_0$: For each pair of vertices $x$ and $y$ where $x
\in (S \cup \TT)$ and $y \in V$, the shortest
path between $x$ and $y$ in $G$ and $G_p$ is stored in $B_0(x,y)$ and $B_0^p(x,y)$ respectively.  The total
size of $B_0$ is $\tilde O((\sigma +\sqrt{n\sigma}) n) = \tilde O(\sigma^{1/2} n^{3/2})$.


\item $B_1$: For each pair of vertices $s \in S$ and $t \in V$, $B_1(s,t)$ contains the
vertex in $\TT$ closest to $t$ on $st$ path, that is $t_s$.  The total size of
$B_1$ is $ O(\sigma n) = \tilde O(\sigma^{1/2}n^{3/2})$.
\end{itemize}

\noindent 
To check if $e(u,v) \in st_s$, we first find $t_s \leftarrow B_1(s,t)$.
Then we check if $B_0^{p}(s,u) +\ B_0^{p}(u,t_s) +\ B_0^{p}(t_s,t) = B_0^{p}(s,t)$ and $B_0^{p}(s,v) +\ B_0^{p}(v,t_s) +\ B_0^{p}(t_s,t)
= B_0^{p}(s,t)$.
If yes, then $e \in st_s$.
 We subdivide the far case into two more sub-cases:
\begin{enumerate}[noitemsep,nolistsep]
   \item The preferred replacement path  avoiding $e$ passes through $t_s$.

   \item The preferred replacement path avoiding $e$ avoids $t_s$.

\end{enumerate}

\noindent  The first case turns out to be a generalization of techniques
used by Demetrescu et. al.\cite{DemetrescuTCR08}  to solve the all pair distance oracle
under single edge/vertex failure -- we will
use the compact version of this algorithm presented by
Pettie and Duan \cite{DuanP09}. The second case is a {\em new and unexplored}
case. We will show that we
can bound the number of preferred replacement paths in this case to $O(\sqrt{n \sigma})$ for a fixed vertex $t$.
This would imply that the total number of such paths is $O(\sigma^{1/2}n^{3/2})$. We are able to
bound the number of paths even though these paths may intersect with each other --  this
is a new feature of our analysis which is much different from the analysis done by
Parter and Peleg \cite{ParterP13} on a related problem.

Section \ref{sec:passes} deals with the first case. In Section \ref{sec:avoids},
we will apply our new approach to the special case when $\sigma =1$,
or there is a single source. In Section \ref{sec:problem}, we will discuss the
potential problems in extending our approach to multiple sources.
Section \ref{sec:multi1} and \ref{sec:multi2}  extends our approach
to multiple sources and in Section \ref{sec:data} we develop our data-structure
that can answer queries in $\tilde O(1)$ time.

\section{Preferred replacement path passes through $t_s$}
\label{sec:passes}
\iflong
Under this assumption, we only need to find the preferred replacement path from $s$ to $t_s$ avoiding $e$. Note that we already know $|t_st|$ length via $B_0(t_{s},t)$. We use the following additional data structures:

\begin{itemize}[leftmargin=*,noitemsep,nolistsep]
\item $B_2$: For $x \in S \cup \TT$ and $y \in V$, $B_2(x,y)$ contains the vertex $w$ on $xy$ path such that $|yw| = 2^{\lfloor \log |xy| \rfloor}$, or in words, the vertex nearest to $x$ whose distance from $y$ is a power of 2. The size of $B_2$ is $\tilde O((\sigma+\sqrt{n\sigma})n) = \tilde O(\sigma^{1/2}n^{3/2})$.

\item $B_3$: For $x \in V$ and $y \in \TT$, $B_4(x,y, \oplus2^i)$
contains the shortest path from $x$ to $y$ avoiding every
$2^i$-th
edge on the path $xy$ from $x$ (where $i \le \log \lfloor
|xy| \rfloor$). Since $|\TT| = \tilde O(\sqrt{n \sigma})$, the size
of $B_3$ is $\tilde O( \sigma^{1/2} n^{3/2} )$.

\item $B_4$: For $s \in S$ and $x \in V$, $B_4(s,x, \ominus2^i)$ contains  the shortest path from $s$ to $x$ avoiding the $2^i$-th
edge on the path $sx$ from $x$ (where $i \le \log \lfloor |sx| \rfloor$). The size of $B_4$ is $ \tilde O( n \sigma) = \tilde O( \sigma^{1/2} n^{3/2} )$.

\item $B_5$: For every $s \in S$ and $x \in \TT$, $B_5(s,x, [\oplus 2^i, \ominus 2^j])$ contains the shortest path from $s$ to $x$ avoiding the sub path
that start from $2^i$-th vertex from $s$ and ends at $2^j$-th
vertex from $x$ on the path $sx$ (where  $i,j \le \log \lfloor sx \rfloor$). The size of $B_5$ is $\tilde O(\sigma^{3/2} n^{1/2}) = \tilde O( \sigma^{1/2} n^{3/2} )$.

\end{itemize}
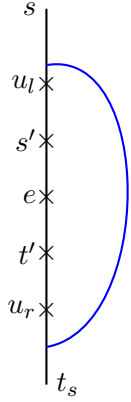
\begin{figure}
\centering

\begin{tikzpicture}[scale=1]

\coordinate (x) at (0,5);
\coordinate (y) at (0,0);
\coordinate (u) at (0,2.5);
\coordinate (ul) at (0,1);
\coordinate (ur) at (0,4);

\coordinate (x1) at (0,1.75);
\coordinate (y1) at (0,3.25);

\draw[thick](x)--(y);
\node[left] at (x){$s$};
\node[right] at (y){$t_{s}$};

\node[left] at (u){$e$};
\node at (u){$\times$};
\node[left] at (ul){$u_r$};
\node at (ul){$\times$};
\node[left] at (ur){$u_l$};
\node at (ur){$\times$};

\node[left] at (x1){$t'$};
\node at (x1){$\times$};
\node[left] at (y1){$s'$};
\node at (y1){$\times$};

\draw[thick,blue] (0,0.5) to[out=10,in=10] node[pos=0.2,
left]
{ } (0,4.25);
\end{tikzpicture}
\caption{The hardest part of the distance oracle, when the
replacement path neither passes through $u_l$ not $u_r$.}
\end{figure}
\noindent   Assume that we get a query
{\sc Q}($s,t,e(u,v)$), that is, find the shortest path from
$s$ to $t$ avoiding $e$. Assume without loss of generality that $u$ is closer to $s$ on $st$ path than $v$. We answer this query as follows:
first we find the distance of $v$ from $s$, via $B_0(s,v)$. If $B_0(s,v)$ is a power of 2 then
we can directly use $B_3(s,t_s, \oplus |sv|) + B_0(t_s,t)$. Else, let $u_l \leftarrow  B_2(s,v)$ and
$u_r \leftarrow B_2(t_s,u)$. Here $u_l$ is the nearest vertex to $s$ on $vs$ path whose
 distance from $v$ is a power of 2. Similarly, $u_r$ is the nearest vertex to $t_s$ on $ut_s$ path whose distance form $u$ is a power of 2.    There are three cases now:

\begin{enumerate}[noitemsep,nolistsep]
\item  The  preferred replacement path passes through $u_l$.

\item  The preferred replacement path passes through $u_r$.

\item The preferred replacement path neither passes through $u_l$
nor $u_r$.

\end{enumerate}
\noindent For the first case, we  can report $B_0(s,u_l) + B_3(u_l, t_{s}, \oplus|u_lv|) + B_0(t_s,t) $. For the second case, we can report $B_4(s,u_r, \ominus  |uu_r|) + B_0(u_r,t_s) + B_0(t_s,t)$ . The hardest
case is when the preferred replacement path neither passes through
$u_l$ nor $u_r$. Let $s'$ be the farthest
vertex  on $sv$ path from $s$ whose distance is a power
of 2, that is $|ss'| = 2^{\lfloor \log |sv| \rfloor}$. Similarly let $t'$ be the farthest vertex on $t_{s}u$
path from $t_s$ whose distance is a power
of 2. Now,
we can  report the shortest path from $s$ to $t_{s}$ avoiding
$[s',t']$ -- this is also stored in $B_5(s,t_s,[\oplus 2^{\lfloor \log |sv| \rfloor}, \ominus 2^{\lfloor
\log |t_su| \rfloor}] )$. By construction, if the shortest path avoids $[u_l,u_r]$, then it also
avoids $[s',t']$. Thus we can return $B_5(s,t_s,[\oplus 2^{\lfloor \log |sv| \rfloor}, \ominus 2^{\lfloor
\log |t_su| \rfloor}] ) +\ B_0(t_s,t)$ as the length  of the preferred replacement path.
The reader can check that the running time of our algorithm
is $O(1)$ as we have to check $B_0,B_1,B_2,B_3,B_4,B_5$ constant number of times.
For the correctness of the above procedure, please refer to \cite{DemetrescuTCR08,DuanP09}.
\fi

Now, we move on to the harder case, that is, replacement paths avoid $t_s$ too. For this, we will  fix a vertex $t$. We will show
that the query
${\sc Q}(s,t,e(u,v))$  can be answered
in $\tilde O(1)$ using $\tilde O(\sqrt{\sigma n})$ space.
This immediately implies that we can answer exact queries
in $\tilde O(1)$ time using $\tilde O(\sigma^{1/2} n^{3/2})$
space.


\section{Preferred Replacement path avoids $t_s$}

\label{sec:avoids}

Handling preferred replacement paths that avoid $t_s$ turns out to be a challenging and unexplored case. For better exposition, we will first solve the problem for the case when $\sigma=1$, that is there is only one source.
Let $\RR$ be the set of all preferred replacement paths from $s$ to $t$ that do not pass through $t_s$. We make  two important observations:

\begin{enumerate}[noitemsep,nolistsep]

\item The size of $\RR$ is $ O(\sqrt n)$.

\item Preferred replacement paths in $\RR$ avoid one contiguous sub path of $st$.

\end{enumerate}

\noindent Few remarks are in order. If the preferred replacement paths in $\RR$ were disjoint, then bounding the size of $\RR$ is easy. However, we are able to bound the size of $\RR$ even if paths are intersecting.  The second observation implies that we can build a balanced binary search tree containing paths in $\RR$. Each node in this tree will contain a preferred replacement path $P$. The key for each node will be the start and end vertex of the sub path $P$ avoids.  We will use this BST to find an appropriate replacement path that avoids an edge $e$.

\begin{definition}(Detour of a replacement path)
Let $P$ be a preferred replacement path avoiding an edge $e$ on $st$ path. Then detour of $P$ is defined as, $\DET(P) := P \setminus st$. That is, detour is a path the leaves $st$ before $e$ till the point it merges back to $st$ again.
\end{definition}

\noindent Since our  replacement path $P$  also avoids $t_s$, the following lemma is immediate by the definition of preferred path.

\begin{lemma}
Let $P$ be a preferred replacement path in $\RR$ that avoids $e$ and $t_s$ on
$st$ path, then (1) $\DET(P)$ cannot merge back to $st_s$ path and (2) $\DET(P)$ is a contiguous path.
\end{lemma}
\begin{lemma}
\label{lem:avoids}
Let $P,P' \in \RR$ avoid $e$ and $e'$ respectively on $st_s$ path. Also assume that $e$ is closer to $s$ than $e'$. Then
(1)  P avoids $e'$  (2) $\DET(P')$ starts  after $e$ on $st_s$ path and (3) $|P| > |P'|$.

\end{lemma}

\iflong
  \begin{proof}

  \noindent (1) $P$ diverges from $st_s$ path above $e$. Since $P
  \in \RR$, it   merges back on $t_st$ path only. Since $e,e'
  \in st_s$ (we are in {\em far} case) and $e$ is closer to $s$ than $e'$, this implies
  that $P$ also avoids $e'$.

  \noindent (2) Assume that  $\DET(P')$ starts above $e$ on
  $st_s$ path. This implies that both $P$ and $P'$ avoid $e$
  and $e'$. But then, our algorithm will choose one of these two paths as a preferred path that avoids both $e$ and $e'$.  Thus,
  we arrive at a contradiction as there are two different preferred replacement paths avoiding $e$ and $e'$.

  \noindent (3) Since $P$  avoids $e'$ (by (1)) and $\DET(P')$ starts after $e$ (by (2)), $P'$ is the preferred path to avoid $e'$ only if $|P'| < |P|$ (else $P$ would be the preferred path as it leaves the $st$ path earlier than $P'$).

  \end{proof}
\fi

\noindent The converse of the third part of the lemma is also true. Since we will be using it in future, we prove it now.

\begin{lemma}
\label{lem:avoidreverse}

Let   $P$ and $P'$ be two preferred replacement paths that avoid $e$ and $e'$ on $st$ path respectively. If $|P| > |P'|$, then $e$ is closer to $s$ than $e'$.
\end{lemma}

\iflong
  \begin{proof}
  Assume for contradiction that $e'$ is closer to $s$ than $e$. Since the replacement path $P'$ has to diverge from $st$ before $e'$ and merge again only in $t_st$,  $P'$ also avoid $e$. But then $P'$ should be the replacement path for avoiding $e$ too, as $|P'| < |P|$, a  contradiction.
  \end{proof}
\fi
\noindent By Lemma \ref{lem:avoids}{\small(3)}, we know that all preferred replacement
paths in $\RR$ have different lengths. In fact, it is the main reason we defined a preferred replacement path. We can thus arrange
these paths   in decreasing order of their lengths. Thus, we get the following corollary.
\begin{corollary}
\label{cor:arrange}
 Given a set $\RR$ of preferred replacement paths from $s$ to $t$ (that also avoid $t_s$), we can arrange  paths in decreasing order of their lengths.
\end{corollary}

\noindent Given a path $P \in \RR$, let $(<P)$ be the set of all preferred replacement paths with length less than $P$. Similarly, let $(>P)$ be the set of all preferred replacement paths with length greater than $P$.   If $P$ avoids $e$, then by Lemma \ref{lem:avoidreverse}, it also avoids all  edges avoided by paths in $(<P)$.  By Lemma \ref{lem:avoids}, for any path $P' \in (<P)$, $\DET(P')$ starts after $e$ on $st_s$ path.
We will now show a simple but important property of a path $P$ in $\RR$.

\begin{lemma}
\label{lem:length}
Let $P \in \RR$ be the shortest path from $s$ to $t$ avoiding $e$ such that $|P| = |st| +\ \ell$ where $\ell \ge 0$, then the size of the set $(<P)$ is $\le \ell$.
\end{lemma}
\iflong
  \begin{proof}
  Since a path in $(<P)$ avoids some edge in $st$ path, its length has to be $\ge |st|$. By Corollary \ref{cor:arrange},
  all paths in $\RR$, and thus $(<P)$ have different lengths. But the length of paths in $(<P)$ is less than the length of $P$.
  Thus, there can be atmost $\ell$ paths in $(<P)$.
  \end{proof}
\fi

\begin{definition}
\label{def:unique}
(Unique path of $P$) Let $\UNQ(P)$ be the prefix of\ $\DET(P)$ which does not intersect with any detours in $\cup_{P' \in (>P)} \DET(P')$.
\end{definition}
We now arrange all  preferred replacement paths in $\RR$ in decreasing order of their lengths. Assume that we are processing a path $P$ according to this ordering such that $P$ avoids $e$ on $st$ path. If $|\UNQ(P)| \ge \sqrt n$, then we have associated $O(\sqrt n)$ vertices on $\UNQ(P)$ to $P$. Else $\UNQ(P) < \sqrt{n}$ and we have the following two cases:

\subsection{   $\DET(P)$ does not intersect with  detour of any path in $(>P)$}
\label{subsec:singlecaseone}
\iflong
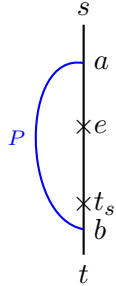
\begin{figure}[hpt!]
\centering
\begin{tikzpicture}[scale=1.7]

\definecolor{dgreen}{rgb}{0.0, 0.5, 0.0}
\begin{scope}[xshift=0cm]
\coordinate (s) at (0,1.8);
\coordinate (t) at (0,0);
\coordinate (ts) at (0,0.4);
\coordinate (b) at (0,.2);

\coordinate (a) at (0,1.5);
\coordinate (v) at (0,1);

\draw[thick](s)--(t);
\node[above] at (s){$s$};
\node[below] at (t){$t$};
\node[right] at (a){$a$};
\node[right] at (b){$b$};

\draw[blue,thick] (a) to[out=170,in=170] node[pos=0.5,left]
{\scriptsize  $P$}  (b);
\node at (v){$\times$};
\node[right] at (v){$e$};

\node at (ts){$\times$};
\node[right] at (ts){$t_s$};
\end{scope}

\end{tikzpicture}

\caption{$\DET(P)$ does not intersect detour of any path in $(>P)$}
\label{fig:singlefirstcase}
\end{figure}
\fi
    Let $\DET(P)$ start at $a$ and end at $b$ -- the vertex where it touches
    $t_st$ path. Let $ab$ denote the
path from $a$ to $b$ on $P$. By our assumption $\UNQ(P) = ab$ and $|ab| < \sqrt n$.
    By Lemma \ref{lem:avoids}, all  replacement paths in $(<P)$ pass through $e$
    (as detour of these replacement paths start below $e$) and by Lemma \ref{lem:avoidreverse},
    these replacement paths avoid edges that are closer to $t$ than $e$. We can view the
    replacement paths as if they are starting from the vertex $a$. That is, consider  paths
    $\{P\setminus sa\} \cup \{ P'\setminus sa | \ P' \in (<P)\}$. These replacement
    paths avoid edges in $at$.   $|P \setminus sa| = |ab| +  |bt| \le |ab|\ + |at| < |at| + \sqrt n$.
    Applying Lemma \ref{lem:length}, we infer that the number of paths in $\{ P'\setminus
sa | P' \in (<P)\}$ is $\le \sqrt n$

\subsection{$\DET(P)$ intersects with  detour of a path in $(>P)$}
\label{subsec:singlecasetwo}
Assume that $P$ first intersects with $P' \in (>P)$. Let $P'$ avoid $e'$ and $\DET(P')$ start at $a'$ and end at $b'$ (see Figure \ref{fig:singlesecondcase}). Let us assume that $\DET(P)$ starts at $a$ and it intersects $\DET(P')$ at $c$. This implies that $\UNQ(P) = ac$.

\noindent Consider the path $sa' \conc a'c \conc ca \conc at$. We claim that this path avoids $e'$. This is due to the fact that by Lemma \ref{lem:avoids}, $\DET(P)$
starts after $e'$ on $st$ path. So, $ca$ and $at$ avoids $e'$. Since $P' = sa' \conc a'c \conc cb' \conc b't$, length of $P'$ must be $\le$ length of the alternate path. Thus,\\
\begin{tabular}{llll}
 & $|sa'| +  |a'c|
+ |cb'| + |b't|$  & $\le$ &$|sa'| + |a'c| + |ca|  + |at|$\\
$\implies$& $
|cb'| + |b't|$ & $\le$ & $|ca|  + |at|$   \\
$\implies$& $|ac| +
|cb'| + |b't|$ & $\le$ &   $2|ca|  + |at|$

\end{tabular}
\begin{figure}[hpt!]
\centering
\begin{tikzpicture}[scale=2]

\definecolor{dgreen}{rgb}{0.0, 0.5, 0.0}
\begin{scope}[xshift=0cm]
\coordinate (s) at (0,1.6);
\coordinate (t) at (0,0);
\coordinate (ts) at (0,0.4);
\coordinate (b1) at (0,.2);

\coordinate (a1) at (0,1.5);
\coordinate (a) at (0,1);
\coordinate (v) at (0,0.7);
\coordinate (v1) at (0,1.2);
\coordinate (c) at (-0.585,0.7);

\draw[thick](s)--(t);
\node[above] at (s){$s$};
\node[below] at (t){$t$};
\node[right] at (a1){$a'$};
\node[right] at (a){$a$};
\node[right] at (b1){$b'$};
\node[left] at (c){$c$};

\draw[blue,thick] (a1) to[out=180,in=180,distance=.8cm] node[pos=0.3,left]
{\scriptsize  $P'$}  (b1);

\draw[red,thick] (a) to[out=180,in=80]
node[pos=0.4,above]
{\scriptsize  $P$}  (c);

\node at (v1){$\times$};
\node[right] at (v1){$e'$};

\node at (v){$\times$};
\node[right] at (v){$e$};

\node at (ts){$\times$};
\node[right] at (ts){$t_s$};
\end{scope}

\end{tikzpicture}

\caption{$\DET(P)$ intersects first with $\DET(P')$ at $c$ where  $P'
\in (>P)$.}
\label{fig:singlesecondcase}
\end{figure}
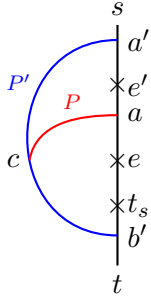
On the left hand of the inequality, we  have a path from $a$ to $t$ avoiding $e$. So, its length should be $\ge$ length of the preferred path $P \setminus sa$. Thus $|P \setminus sa| \le  2|ca|  + |at| \le 2\sqrt n\
+\ |at|$. By Lemma \ref{lem:avoids},
all  replacement paths in $(<P)$ pass through $e$ (as
detour of these replacement paths start below $e$) and by
Lemma \ref{lem:avoidreverse}, these replacement paths avoid
edges that are closer to $t$ than $e$. We can view the
replacement paths as if they are starting from the vertex $a$.
That is, consider  paths $\{P\setminus sa\} \cup \{ P'\setminus
sa | P' \in (<P)\}$. Applying lemma \ref{lem:length}, we infer that the number of
paths in $\{ P'\setminus
sa | P' \in (<P)\}$ is $\le 2 \sqrt n$.

Our arguments above point to the following important observation: {\em Once we find a replacement path in $\RR$ with unique path length $< \sqrt n$, then there are at most 2$\sqrt n$ replacement paths in $\RR$ left to process.}
Since there can be at most $\sqrt n$ paths in $\RR$ with unique path length $\ge \sqrt n$, we have proven the following lemma:

\begin{lemma}
\label{lem:sizeR}
$|\RR| = O(\sqrt n)$.
\end{lemma}

We now build a data-structure which will exploit Lemma \ref{lem:sizeR}.
However, we need another key but simple observation. By Lemma \ref{lem:avoids}, if $|P| > |P'|$, then $\DET(P')$ starts below the edge avoided by $P$. This lemma implies that $\DET(P')$ starts below all  edges avoided by $P$. Thus $P$ avoids some contiguous path in $st_s$ and detour of all replacement paths in $(<P)$ start below the last edge (which is closer to $t_s$) in this subpath. Thus, we have proved the second key lemma:

\begin{lemma}
\label{lem:contiguous}

A replacement path $P$ avoids a contiguous subpath of $st$.
\end{lemma}
Let $\FF(P)$ and $\LL(P)$ denote the first and the last vertex of the contiguous path that $P$ avoids. Given a vertex $v$, let $v.depth$ denote the depth of $v$ in the BFS tree of $s$.  We can store the depth of all  vertices in an array (takes $O(n)$ space). Lastly, we build a balanced binary search tree\ BST($t$) in which each node represents a path $P$. The key used to search the node is the range: $[\FF(P).depth, \LL(P).depth]$. By Lemma \ref{lem:contiguous}, all  replacement paths avoid  contiguous subpaths of $st_s$.  These contiguous paths are also disjoint as there is only one preferred path avoiding an edge. Thus, the key we have chosen forms a total ordered set with respect to the relation $\{ <,>\}$.  The size of BST($t$) is $O(\sqrt n)$ as the size of $\RR$ is $O(\sqrt n)$. We are now ready to process any query {\sc Q}$(s,t,e(u,v))$. We just need to search for an interval in BST($t$) that contains $u.depth$ and $v.depth$. This can be done in $\tilde O(1)$ time. Thus we have proved the following theorem:

\begin{theorem}
There exists a data-structure of size $\tilde O(n^{3/2})$ for single source single fault tolerant exact distance oracle that can answer each query in $\tilde O(1)$ time.
\end{theorem}

\section{From Single Source to Multiple Sources}

Unfortunately, the analysis for the single source case is not easily extendible to multiple source case. We identify the exact problem here. Consider the case described in Section \ref{subsec:singlecasetwo}. In this case, we show that if $|P'| > |P|$ and $P$ intersects with $P'$, then there is a path available for $P$ (that is $ac \conc
cb' \conc b't$).  We can use this path because it also avoids $e $ (the edge avoided by path $P$). First, we show that the above assertion is not true when we move to multiple source case. Consider the following example (See Figure \ref{fig:multiple}).
Here, $P$ avoids $e$ on $st$ path and $P'$ avoids $e'$ on $s't$ path. $\DET(P)$ starts at $a$ and its intersects $P'$ at $c$. $\DET(P')$ starts at $a'$ and it hits $st$ path at $b'$ and then   passes through $e$. Note that the full path $P$ from $s$ to $t$ is not shown in Figure \ref{fig:multiple}. The reader can check that the  path $ac \conc
cb' \conc b't$ is not an alternate path for $P$ as it passes through $e$. We say that such a path is a bad path because it breaks the easy analysis of single source case (we will formally define bad paths in Section \ref{subsec:pintersects}).
However, we are able to show that the total number of {\em good} paths (paths which are not bad)
is $\ge$ the number of bad paths. Good paths exhibit properties similar to the set
$\RR$ in Section \ref{sec:avoids}. This will help us in bounding them (and thus bad paths too).
\label{sec:problem}
\iflong
\begin{figure}[hpt!]

\centering
\begin{tikzpicture}[scale=1.5]

\definecolor{dgreen}{rgb}{0.0, 0.5, 0.0}
\begin{scope}[xshift=0cm]
\coordinate (s) at (-1,2);
\coordinate (s1) at (1,2);
\coordinate (t) at (0,0);
\coordinate (ts) at (0,0.4);
\coordinate (b1) at (0.5,1);

\coordinate (a1) at (-0.75,1.5);
\coordinate (a) at (0.65,1.3);
\coordinate (v) at (0.34,0.7);
\coordinate (v1) at (-0.6,1.2);
\coordinate (c) at (-0.2,1.12);

\draw[thick](s)--(t);
\draw[thick](s1)--(t);
\node[above] at (s){$s'$};
\node[above] at (s1){$s$};
\node[below] at (t){$t$};
\node[left] at (a1){$a'$};
\node[right] at (a){$a$};
\node[right] at (b1){$b'$};
\node[below] at (c){$c$};

\draw[blue,thick] (a1) to[out=330,in=180,distance=.8cm]
node[pos=0.3,above]
{\scriptsize  $P'$}  (b1);

\draw[red,thick] (a) to[out=180,in=80]
node[pos=0.4,above]
{\scriptsize  $P$}  (c);

\node at (v1){$\times$};
\node[left] at (v1){$e'$};

\node at (v){$\times$};
\node[right] at (v){$e$};

\end{scope}

\end{tikzpicture}

\caption{The bad case for us: $P' \in (>P)$ intersects with $P$ and then passes through the edge $e$ that $P$ avoids. }
\label{fig:multiple}
\end{figure}
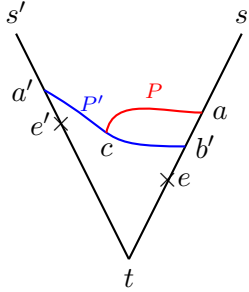
\fi
Once again we will fix a vertex $t$ and show that  the number of replacement paths from
$s \in S$ to $t$ that also avoids $t_s$ is $ O(\sqrt {\sigma n})$.
Let $\BFS(t)$ denote  the  union of all  shortest paths from $t$ to $s \in S$.
The reader can check that the union of these paths does not admit a cycle, so we
can assume that its a tree rooted at $t$. Since $\BFS(t)$ has at most $\sigma$
leaves, the number of vertices with degree $> 2$ in $\BFS(t)$ is $O(\sigma)$.
We now contract all the vertices of degree 2 (except $t$ and $s \in S$) in $\BFS(t)$
to get a tree that only contains
leaves of $\BFS(t)$, the root $t$, all the sources and all other vertices
with degree $> 2$ in $\BFS(t)$.
\begin{definition} (\SBFS($t$))
\SBFS($t$) is a tree obtained by contracting all the vertices with degree exactly 2 in $\BFS(t)$
except $t$ and source $s \in S$.
\end{definition}

\begin{definition} (Intersection vertex and segment in \SBFS($t$))

\noindent Each node  \SBFS$(t)$  is called an intersection vertex.
An edge $xy \in$ \SBFS($t$) denotes a path between two vertices in $\BFS(t)$.
We call such an edge in \SBFS\  a segment. We use this term in order to differentiate between  edges in $\BFS(t)$ and  \SBFS($t$).
Also, we will use the following convention: if $xy$ is a segment,
then $y$ is closer to $t$ than $x$.
\end{definition}
\noindent \SBFS($t$) has at most $\sigma$
vertices with degree $\le 2$. This implies that there are at most $O(\sigma)$
intersection vertices and segments in \SBFS($t$).

As in the single source case, we would like to find the preferred path for each avoided edge on the $st$
path where $s \in S$. However, we don't have enough space to store all these paths.
Also storing all  paths seems wasteful.
Consider two preferred replacement paths $P$ and $P'$ that start from $s$ and $s'$ respectively.
These two paths meet at an intersection vertex $x$ after which they are same, that is, they take
the same detour to reach $t$. Storing both $P$ and $P'$ seems wasteful as they are essentially
the same path once they hit $x$.  To this end, we only store preferred path
corresponding to each segment in \SBFS($t$). We now describe our approach in detail.

Let $xy$ be a segment in \SBFS($t$). We divide replacement paths whose detour start in $xy$ into
two types:
\begin{itemize}[noitemsep,nolistsep]

   \item[] $\TON(xy)$: Preferred replacement paths from $x$ to $t$ whose detour starts in $xy$
    but the avoided edge lies in $yt_x$.

   \item[] $\TTW(xy)$: Preferred replacement paths from $x$ to $t$ whose
   start of detour and avoided edge both lie strictly inside segment $xy$ (that is, detour cannot
   start from $x$ or $y$).
\end{itemize}

\noindent Let $\TON := \cup_{xy \in \text{\SBFS($t$)}} \TON(xy)$ and
$\TTW := \cup_{xy \in \text{\SBFS($t$)}} \TTW(xy)$. The set  $\TON$ helps us to weed out
simple preferred replacement paths.
We will show that we can store  preferred replacement paths in  $\TON$ in
$O(\sigma)$ space --  one per segment in \SBFS($t$).
The hardest case for us in $\TTW$, which contains bad paths.
Let $\BP$ denote the set of bad paths in $\TTW$. We will show  that
$|\BP| \le |\TTW \setminus \BP|$ (the number of bad paths is $\le$ number of good paths in $\TTW$)
and $|\TTW \setminus \BP| =  O(\sqrt{n\sigma})$ (the number of good path is $ O(\sqrt{n\sigma}$)).
This implies that $|\TTW| =  O(\sqrt{n \sigma})$. 

Since $\TON$ and $\TTW$ are of size $O(\sqrt{n\sigma})$, we can make a data-structure
of size $O(\sqrt{n\sigma})$.
In this data-structure, we have stored a preferred path for each segment.
However, we have to answer queries of type $\textsc{Q}(s,t,e)$
where $s$ is a source.
In Section \ref{sec:data}, we will see
how to use preferred paths of segments to
answer queries in $\tilde O(1)$ time.

\section{Analysing preferred replacement paths in $\TON$}
\label{sec:multi1}
We first show the following:
\begin{lemma}
\label{lem:allsame}
For each segment $xy \in$ \SBFS$(t)$, $|\TON(xy)| = 1$
\end{lemma}
\iflong
\begin{proof}
  Assume that there are two preferred paths
   $P$ and $P'$ whose detour start in $xy$ and their avoided edge $e$ and
  $e'$(respectively) lie in $yt_x$. Since we are analyzing
  paths in the {\em far case},  detours of $P$ and $P'$ meet $xt$ path in $t_xt$ subpath.
  Since both $e$ and $e'$ lie on $yt_x$ path,
  this implies that both $P$ and $P'$ avoid $e$ and $e'$.
  Thus, we will choose the smaller of the two paths as the preferred path avoiding both
  $e$ and $e'$. Else if $|P| = |P'|$, then we will choose that path which leaves the $xt$
  path as early as possible. Thus, there is only one preferred path avoiding both
  $e$ and $e'$, a contradiction.
  \end{proof}
\fi

\noindent The above lemma implies that $|\TON| = \cup_{xy \in \text{\SBFS($t$)}} |\TON(xy)| =
O(\sigma) = O(\sqrt{n \sigma})$.

\section{Analysing preferred replacement paths in $\TTW$}
\label{sec:multi2}


We first show that one special kind of path will never lie in $\TTW$.
This characterization will help in analyzing bad paths in $\TTW$.

\begin{lemma}
\label{lem:feature}
Let $P$ be a preferred path from $x$ to $t$ avoiding $e$ on $xt$ path.
If $P$ merges with any segment $x'y'$ and then diverges from $x't$ path, then $P \notin \TTW$.
\end{lemma}
\iflong
  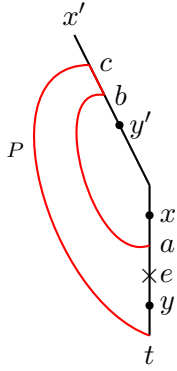
\begin{figure}[hpt!]
\centering
\begin{tikzpicture}[scale=2]
\begin{scope}
\coordinate (s) at (-0.5,2);
\coordinate (s1) at (0.5,2);
\coordinate (w) at (0,1);
\coordinate (t) at (0,0);
\coordinate (x) at (0,0.8);
\coordinate (y) at (0,0.2);
\coordinate (y1) at (-0.2,1.4);
\coordinate (i1) at (-0.4,1.8);
\coordinate (i2) at (-0.3,1.6);
\coordinate (a) at (0,0.6);
\coordinate (v) at (0,0.4);

\draw[thick](s)--(w);
\draw[thick](w)--(t);
\node[above] at (s){$x'$};
\node[below] at (t){$t$};

\node[right] at (a){$a$};

\node[right] at (i1){$c$};
\node[right] at (i2){$b$};

\draw[red,thick] (a) to[out=200,in=170]
(i2);
\draw[red,thick](i1)--(i2);
\draw[red,thick] (i1) to[out=180,in=160]node[pos=0.4,left,black]{\scriptsize{$P$}}  (t);
\node at (v){$\times$};
\node[right] at (v){$e$};
\draw (y) node[fill,circle,scale=0.3]{};
\node[right] at (y){$y$};
\draw (x) node[fill,circle,scale=0.3]{};
\node[right] at (x){$x$};
\draw (y1) node[fill,circle,scale=0.3]{};
\node[right] at (y1){$y'$};
\end{scope}
\end{tikzpicture}
\caption{The path $P$ merges with another segment $x'y'$  and then diverges.}
\label{fig:feature}

\end{figure}
  \begin{proof}
  Once $P$ merges with $x'y'$ segment, it will diverge from it only if
  $x't$ contains $e$. This implies that $xt$ and $x't$ intersect or $xt$
  is a subpath of $x't$.

  \noindent Assume for contradiction that $P \in \TTW$, that is, $\DET(P)$
  start strictly inside segment $xy$. Consider the Figure \ref{fig:feature} in which
  the detour of the preferred replacement path $P$
  starts after $x$ (at $a$). It then intersect  $x'y'$ at $b$
  and then diverges from $x't$ path at $c$.
  We claim that there exists another path to reach $b$ that is shorter  than
  $xa \conc ab$. This path is $xb$, where path $xb$ is a subpath of $x't$.
  It is also a shorter path (in $G_p$) as it uses
   edges in original $x't$ path.

  This path (1) diverges at $x$ and (2) is shorter (in $G_p$) as it uses edges from $x't$ path.
  Thus, there is a shorter replacement path than $P$ that avoids $e$. This path is not in
  $\TTW$ as its detour starts at $x$  .
  This leads to a contradiction as we had assumed that $P$ is the preferred replacement path avoiding $e$.
  Thus our assumption,
  namely that $P \in \TTW$ must be false.
  \end{proof}
\fi

\noindent We will now analyze  paths in $\TTW$. Consider two replacement
paths $P, P'$ avoiding edges $e,e'$ (respectively) on $xy,x'y'$ segment respectively.
Let  $a,a'$ be the starting vertex of $\DET(P), \DET(P')$ respectively.
We say that $P \prec P'$ if $|at| < |a't|$.
If $|at| = |a't|$, then the tie is broken arbitrarily.

\noindent  Given a path $P \in \TTW$, let $(< P)$ be the set of all replacement
paths in $\TTW$ that are $\prec P$ in the ordering. Similarly, $(> P)$ is
the set of all replacement paths $P' \in \TTW$ for which $P \prec P'$.
Define $\UNQ(P)$ according to this ordering (see definition \ref{def:unique}).
Assume that we are processing a replacement path $P$ according to this ordering.
If $|\UNQ(P)| \ge \sqrt{n/\sigma}$, then we can associate $O(\sqrt{n/\sigma})$
{\em unique} vertices to $P$. Otherwise $|\UNQ(P)| < \sqrt{n/\sigma}$ and we have the following two cases:

\subsection{ $\DET (P)$ does not intersect with any other detour in $(> P)$ }
\label{subsec:nointersect}
\noindent This case is similar to the first case in Section \ref{subsec:singlecaseone}.
\iflong
  Assume that $P$ avoids an edge $e$ on segment $xy$. Let $\DET(P)$
  start at $a \in xy$ and end at $b$ -- the vertex
  where it touches $t_xt$ path. Let $ab$ denote the
path from $a$ to $b$ on $P$. By our assumption $\UNQ(P)
  = ab$ and $|ab| < \sqrt{n/\sigma}$. Consider the following set
  of replacement paths  $(< P)_x := \{P' \in (< P)\ |\ P'~ \text{avoids an edge on $xy$ segment}\}$.
  By Lemma \ref{lem:avoids},
  all  replacement paths in $(< P)_x$ pass through $e$ (as
  detour of these replacement paths start below $e$) and by
  Lemma \ref{lem:avoidreverse}, these replacement paths avoid
  edges that are closer to $y$ than $e$. We can view these
  replacement paths as if they are starting from vertex $a$.
   $|P \setminus xa|
  = |ab| + |bt| \le \sqrt {n/\sigma} + |at|$. Using lemma \ref{lem:length}, total
  number of paths in   $(\le P)_x$ is  $\le \sqrt{ n/\sigma}$.  Thus, once we
  get a replacement path $P \in \TTW(xy)$ with $|\UNQ(P)| < \sqrt{n /\sigma}$, then there
  are at most $\sqrt{n/\sigma}$ replacement paths in $\TTW(xy)$ remaining to be processed.
  Thus, total number of paths in $\TTW$ with  $|\UNQ(P)|< \sqrt{n/\sigma}$ is
  $\sum_{xy \in \text{\SBFS($t$)}} \sqrt{n/ \sigma} = O(\sqrt{n \sigma})$
  (as there are $O(\sigma)$ segments in \SBFS($t$))
\fi

\subsection{ $\DET(P)$ intersects with  detour of a  path in $(> P)$}
\label{subsec:pintersects}

We first give a formal definition of a bad path that was defined informally in Section \ref{sec:problem}.
\begin{definition} (Bad Path)
A path $P \in \TTW$ is called a bad path if there exists another path $P' \in (>P)$
such that (1) $\DET(P)$ intersects with $\DET(P')$ and (2) $\DET(P')$ passes through the edge
avoided by $P$ after their intersection. We also say that $P$ is a bad replacement
path due to $P'$ if $P'$ satisfies the above two conditions.
\end{definition}

A path that is not bad is called a good path. In Section \ref{sec:problem}, we saw that
bad paths break the easy analysis of the single source case.
So, we have two cases depending on whether the path is good or bad. Let us look at the easier case first.\iflong\\\fi

\iflong\iflong
\begin{figure}[hpt!]
\else
\begin{wrapfigure}[]{r}{.4\textwidth}
\fi
\centering
\begin{tikzpicture}[scale=2]

\definecolor{dgreen}{rgb}{0.0, 0.5, 0.0}
\begin{scope}[xshift=0cm]
\coordinate (s) at (0,2);
\coordinate (t) at (0,0);
\coordinate (ts) at (0,0.4);
\coordinate (b1) at (0,.2);
\coordinate (y) at (0,1);
\coordinate (a1) at (0,1.8);
\coordinate (a) at (0,1.4);
\coordinate (v) at (0,1.2);
\coordinate (v1) at (0,1.6);
\coordinate (c) at (-0.558,0.7);

\draw[thick](s)--(t);
\node[above] at (s){$x=x'$};
\node[below] at (t){$t$};
\node[right] at (a1){$a'$};
\node[right] at (a){$a$};
\node[right] at (b1){$b'$};
\node[left] at (c){$c$};

\draw[blue,thick] (a1) to[out=180,in=180,distance=.8cm]
node[pos=0.3,left]
{\scriptsize  $P'$}  (b1);

\draw[red,thick] (a) to[out=180,in=80]
node[pos=0.4,above]
{\scriptsize  $P$}  (c);

\node at (v1){$\times$};
\node[right] at (v1){$e'$};

\node at (v){$\times$};
\node[right] at (v){$e$};

\node at (ts){$\times$};
\node[right] at (ts){$t_x$};

\draw (y) node[fill,circle,scale=0.3]{};
\node[right] at (y){$y$};

\end{scope}

\begin{scope}[xshift=1.7cm]
\coordinate (s) at (-0.5,2);
\coordinate (s1) at (0.5,2);

\coordinate (t) at (0,0);
\coordinate (ts) at (-0.15,0.6);
\coordinate (b1) at (-0.1,0.4);
\coordinate (y1) at (-0.32,1.3);
\coordinate (y) at (0.32,1.3);
\coordinate (a1) at (-0.45,1.8);
\coordinate (a) at (0.44,1.8);
\coordinate (v) at (0.4,1.6);
\coordinate (v1) at (-0.395,1.6);
\coordinate (c) at (-0.65,0.66);

\draw[thick](s)--(t);
\draw[thick](s1)--(t);
\node[above] at (s){$x'$};
\node[above] at (s1){$x$};
\node[below] at (t){$t$};
\node[right] at (a1){$a'$};
\node[right] at (a){$a$};
\node[right] at (b1){$b'$};
\node[left] at (c){$c$};

\draw[blue,thick] (a1) to[out=180,in=180,distance=.8cm]
node[pos=0.3,left]
{\scriptsize  $P'$}  (b1);

\draw[red,thick] (a) to[out=180,in=10]
node[pos=0.9,above]
{\scriptsize  $P$}  (c);

\node at (v1){$\times$};
\node[right] at (v1){$e'$};

\node at (v){$\times$};
\node[right] at (v){$e$};

\node at (ts){$\times$};
\node[right] at (ts){$t_{x'}$};

\draw (y1) node[fill,circle,scale=0.3]{};
\node[right] at (y1){$y'$};

\draw (y) node[fill,circle,scale=0.3]{};
\node[right] at (y){$y$};

\end{scope}

\end{tikzpicture}

\caption{The figure shows two representative examples when $ca$ and $at$ does not pass through $e'$.}
\label{fig:examples}
\iflong
\end{figure}
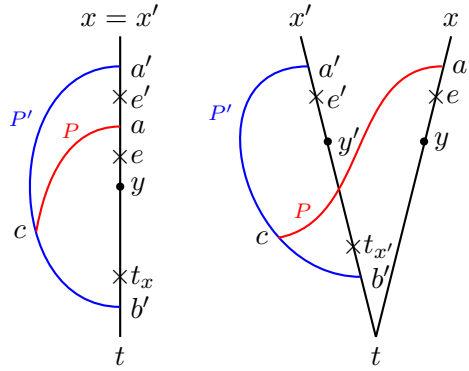
\else
\end{wrapfigure}
\fi\fi

\noindent {(1) \em $P$ is a good path.}

\noindent Assume that $P \in \TTW(xy)$ and it avoids an edge $e \in xy$.
Assume that $P$  intersects first with $P' \in (> P)$ and $P'$ avoids $e'$ on $x'y'$ segment.
Note that $x$ may be equal to $x'$.
Let $\DET(P')$ start at $a'$ and end at
$b'$. Assume that $\DET(P)$ starts at $a$ and it intersects $\DET(P')$
at $c$.
Consider the path $x'a' \conc a'c \conc ca \conc at$. Since $x'a' \conc a'c$ is a part of $P'$,
it avoids $e'$. However, it is not clear whether $ca \conc at$ avoids $e'$ too.
In Figure \ref{fig:examples}, we see two representative examples in which $ca$ and $at$
avoid $e'$.
\iflong
  We will now show that both $ca$   and $at$ cannot passes through $e'$.
  \begin{enumerate}
   \item[(a)] Assume that  $ca$ passes through $e'$ 
    \label{enum:goodcase}

  If $x=x'$, then by Lemma \ref{lem:avoids}, as $P' \in (>P)$, $\DET(P)$ (and hence $ca$)
   starts below $e'$. Thus, $ca$ cannot pass through $e'$ as $\DET(P)$ does not
  intersect $xt_x$ path and $e' \in xt_x$. So let us assume that $x \neq x'$.
  This implies that $P$ intersects $x'y'$. After intersecting $x'y'$ path
    $P$ did not follow $x'y' \conc y't$ (since $ca$ intersect $\DET(P')$ at $c$).
    This implies that $P$ intersect with another path $x'y'$ and then
    diverges. By Lemma \ref{lem:feature}, $P \notin \TTW$.
    This leads to a contradiction as we assumed that $P \in \TTW$. Thus our
    assumption, namely $ca$ passes through
    $e'$ is false.



  \item[(b)] Assume that $at$ passes through $e'$ 
  \label{enum:one}

  If $x = x'$, then by Lemma \ref{lem:avoids}, starting vertex of $\DET(P)$, that is $a$, starts below $e'$. Thus,
  $at$ cannot pass through $e'$. So let us assume that $x \neq x'$.
  If $at$ passes through $e'$, then segment $x'y'$ is a subpath of $xt$ path.
  This is due to the fact that $a \in xy$ and $e' \in x'y'$. Since $P' \in \TTW(x'y')$,
  $\DET(P')$ starts strictly inside segment $x'y'$, at vertex $a'$. This implies that $|a't| < |at|$.
  This contradicts our assumption that $P' \in (>P)$. Thus, $at$ cannot pass through $e'$.


  \end{enumerate}
\fi


Thus, the path $x'a' \conc a'c \conc ca \conc at$ is indeed a valid replacement
path from $x'$ to $t$ avoiding $e'$. Since $P' = x'a' \conc a'c \conc cb' \conc b't$, length
of $P'$ must be $\le$ length of  this alternate path. Thus,\\
\begin{tabular}{llll}
& $|x'a'| +  |a'c|
+ |cb'| + |b't|$& $\le$ & $|x'a'| + |a'c| + |ca|  + |at|$  \\
$\implies$& $
|cb'| + |b't|$ & $\le$ & $|ca|  + |at|$   \\
$\implies$&$|ac| +
|cb'| + |b't|$ & $\le$ &  $2|ca|  + |at|$

\end{tabular}

On the left hand of the inequality, we  have a path from
$a$ to $t$ avoiding $e$ (since we know that $P$ is a good path, so $P'$ and thus $cb' \conc b't$ does not
pass through $e$). So, its length should be $\ge$
length of the preferred path $P \setminus xa$. Thus $|P
\setminus xa| \le 2|ca|  + |at| \le 2\sqrt {n/\sigma}\
+\ |at|$ (since $|\UNQ(P)| = |ac| < \sqrt{n/\sigma}$). Consider the following set of replacement paths  $(< P)_x
:= \{P' \in (< P)\ |\ P' \ \text{avoids
an edge on $xy$ segment}\}$. By Lemma \ref{lem:avoids},
all  replacement paths in $(<P)_x$ pass through $e$ (as
detour of these replacement paths start below $e$) and by
Lemma \ref{lem:avoidreverse}, these replacement paths avoid
edges that are closer to $y$ than $e$.
  Applying Lemma \ref{lem:length}, we get
that  the number of replacement paths $(<P)_x$  is $ \le 2\sqrt{n/\sigma}$. Thus, once we
get a replacement path $P \in \TTW(xy)$ with $|\UNQ(P)| < \sqrt{n /\sigma}$, then there
are at most $2\sqrt{n/\sigma}$ replacement paths in $\TTW(xy)$ remaining to be processed.
Thus, total number of paths $\in \TTW$ with  $|\UNQ(P)| < \sqrt{n/\sigma}$ is
$\sum_{xy \in \text{\SBFS($t$)}} 2\sqrt{n/\sigma} = O(\sqrt{n \sigma})$
(as there are $O(\sigma)$ segments in \SBFS($t$)).\iflong\\\fi

\iflong
   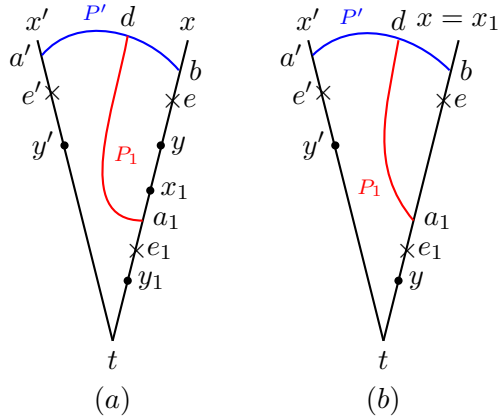
\begin{figure}[hpt!]
\centering
\begin{tikzpicture}[scale=2]
\begin{scope}
\coordinate (s) at (0.5,2);
\coordinate (s1) at (-0.5,2);
\coordinate (t) at (0,0);
\coordinate (ts) at (0.15,0.4);
\coordinate (b) at (0.44,1.8);
\coordinate (y1) at (-0.32,1.3);
\coordinate (y) at (0.32,1.3);

\coordinate (x2) at (0.25,1);
\coordinate (y2) at (0.1,0.4);

\coordinate (a1) at (-0.47,1.9);
\coordinate (a2) at (0.2,0.8);
\coordinate (v) at (0.4,1.6);
\coordinate (v1) at (-0.4,1.65);
\coordinate (v2) at (0.16,0.6);
\coordinate (c) at (-0.7,0.66);
\coordinate (d) at (0.1,2.03);

\draw[thick](s)--(t);
\draw[thick](s1)--(t);
\node[above] at (s){$x$};
\node[above] at (s1){$x'$};
\node[below] at (t){$t$};
\node[left] at (a1){$a'$};
\node[right] at (b){$b$};

\node[above] at (d){$d$};

\draw[blue,thick] (a1) to[out=50,in=130]
node[pos=0.4,above]
{\scriptsize  $P'$}  (b);

\draw[red,thick] (a2) to[out=180,in=260]
node[pos=0.5,right]
{\scriptsize  $P_1$}  (d);

\node at (v1){$\times$};
\node[left] at (v1){$e'$};

\node[right] at (a2){$a_1$};

\node at (v){$\times$};
\node[right] at (v){$e$};

\node at (v2){$\times$};
\node[right] at (v2){$e_1$};


\draw (y1) node[fill,circle,scale=0.3]{};
\node[left] at (y1){$y'$};

\draw (y) node[fill,circle,scale=0.3]{};
\node[right] at (y){$y$};

\draw (x2) node[fill,circle,scale=0.3]{};
\node[right] at (x2){$x_1$};

\draw (y2) node[fill,circle,scale=0.3]{};
\node[right] at (y2){$y_1$};

\node at (0,-0.4){$(a)$};

\end{scope}
\begin{scope}[xshift=1.8cm]
\coordinate (s) at (-0.5,2);
\coordinate (s1) at (0.5,2);

\coordinate (t) at (0,0);
\coordinate (ts) at (0.1,0.4);
\coordinate (b) at (0.44,1.8);

\coordinate (a2) at (0.2,0.8);
\coordinate (d) at (0.1,2);
\coordinate (y1) at (-0.32,1.3);
\coordinate (y) at (0.1,0.4);
\coordinate (a1) at (-0.47,1.9);
\coordinate (a) at (0.1,0.8);
\coordinate (v) at (0.4,1.6);
\coordinate (v2) at (0.16,0.6);
\coordinate (v1) at (-0.4,1.65);
\coordinate (c) at (-0.7,0.66);

\draw[thick](s)--(t);
\draw[thick](s1)--(t);

\node[above] at (s){$x'$};
\node[above] at (s1){$x=x_1$};
\node[below] at (t){$t$};
\node[left] at (a1){$a'$};

\node[right] at (b){$b$};
\node[right] at (a2){$a_1$};
\node[above] at (d){$d$};

\draw[blue,thick] (a1) to[out=45,in=135]
node[pos=0.3,above]
{\scriptsize  $P'$}  (b);

\draw[red,thick] (a2) to[out=130,in=260]
node[pos=0.2,left]
{\scriptsize  $P_1$}  (d);

\node at (v1){$\times$};
\node[left] at (v1){$e'$};

\node at (v2){$\times$};
\node[right] at (v2){$e_1$};

\node at (v){$\times$};
\node[right] at (v){$e$};

\draw (y1) node[fill,circle,scale=0.3]{};
\node[left] at (y1){$y'$};

\draw (y) node[fill,circle,scale=0.3]{};
\node[right] at (y){$y$};

\node at (0,-0.4){$(b)$};
\end{scope}
\end{tikzpicture}
\caption{Two cases in which $P'$ passes through $e_1$ after intersecting with $P_1$}
\label{fig:badexample}

\end{figure}
\fi
\noindent  {(2) \em $P$ is a bad path.}
\label{enum:two}

\noindent We now arrive at our hardest scenario.
We will first show that the number of good paths in $\TTW$
is greater than
the number of bad paths in $\TTW$.
To this end, we will prove the following lemma:
\begin{lemma}
\label{lem:badpaths}
For each $P' \in \TTW$, there exists only one replacement
path $P \in \TTW$ which is bad
due to $P'$.
\end{lemma}

\iflong
  \begin{proof}

  Assume that $P$ is the preferred path
  from $x$ to $t$ avoiding $e$ on segment $xy$. Similarly, $P'$ is the preferred path from $x'$ to $t$ avoiding $e'$ on segment $x'y'$ and $P$ is bad due to $P'$. Assume for contradiction that there is one more replacement
  path $P_1$ which is bad
  due to $P'$.
  Let us assume that $P_1$ avoids $e_1$ on $x_1y_1$ segment.
  If $x' = x$ or $ x'= x_1$, then $\DET(P')$ cannot pass through
  $e$ or $e_1$ respectively as
  $\DET(P')$ starts before $e$ or $e_{1}$ (as $P' \in (>P)$ or $P' \in (>P_1)$) and touches $x't$ path only at $t_{x'}t$. So, let us assume
  that $x' \neq x$ and $ x' \neq x_1$.

  Since $P' \in \TTW$, by lemma \ref{lem:feature}, we know
  that it follows $xt$ path after hitting segment $xy$, say at $b$. This
  implies
  that $e_1$ also lies on the $xt$ path. Thus, $xt$ and $x_1t$
  intersect. Without loss of generality, assume that  the segment $x_1y_1$ is
  a subpath of $xt$.
  Let us assume that $\DET(P_1)$ starts at $a_1$ and it
  hits $P'$ at $d$.
  There are two ways for $P_1$ to reach $d$ (both avoiding
  $e_1$):
  $x_1a_1 \conc a_1d$ and $x_1b \conc bd$
  where the first path is using the detour of $P_1$ and the second path
  uses $xt$ path to reach $b$.

  Note that the second path leaves the $x_1t$ path earlier than the
  first path ($x_1$ compared to $a_1$ if $x \neq x_1$ (Figure  \ref{fig:badexample}(a)) and
  $b$ compared to $a_1$ if $x=x_1$ (Figure  \ref{fig:badexample}(a))). Even then the preferred
  path used the first alternative.
  This implies that the length of the first path
  must be {\em strictly} less than the second.  Thus,
  $|x_1a_1| + |a_1d| < |x_1b| + |bd|$. Thus,
  \begin{equation}
  |a_1d| < |db| + |bx_1| - |x_1a_1|
  \end{equation}
  $P'$ takes the following path
  $x'a' \conc a'b \conc bt$.
  But there is  an alternative path available for $P'$, it
  is $x'a' \conc a'd \conc da_1 \conc a_1t$.
  The path is a valid path avoiding $e'$ only if $a_1d$
  does not pass through $e'$. All other components of this
  path are a part of $P'$ ($a'd \in a'b$ and $a_1t \in bt$)
  .

  If $a_1d$ does not pass through $e'$, then we can show that
  the alternative path has a length less than $|P'|$, thus arriving
  at a contradiction.
  Consider the length of the alternative path:
  \begin{tabbing}
   $|x'a'| $\=$+ |a'd| + |da_1| + |a_1t|$\\
   \>$< |x'a'|+ |a'd| + |db|  + |bx_1| - |x_1a_1| + |a_1t|$
  \hspace{10mm}\ (Using  Equation (1))\\\\
  If $x \neq x'$(See Figure  \ref{fig:badexample}(a)), then
  $ |bx_1|  - |x_1a_1|
  + |a_1t| \le  |bx_1|  + |x_1a_1| + |a_1t| = |bt|$. Else if $x
  =x'$ \\
  (See Figure  \ref{fig:badexample}(b)), then $ |bx_1|  - |x_1a_1|
  + |a_1t| = -|ba_1| +\ |a_1t| \le|ba_1| +\ |a_1t| =  |bt|$ \\\\
   \>$\le |x'a'|+ |a'd| + |db| + |bt|$\\
  \>$= |x'a'|+ |a'b| + |bt|$\\
   \>$= |P'|$
  \end{tabbing}

  This leads to a contradiction as we have assumed that $P'$
  is
  the shortest path from $x'$ to $t$ avoiding $e'$.

  To end this proof, we will show that $a_1d$ cannot pass
  through $e'$.
  Assume for contradiction that $a_1d$ passes through $e'$.
  This
  implies that $\DET(P_1)$ intersects with $x'y'$ segment (as
  $e' \in x'y'$) and
  then diverges from it (as $\DET(P_1)$ intersect with $\DET(P')$
  at $d$).
  By Lemma \ref{lem:feature}, $P_1 \notin \TTW$. But this cannot
  be the
  case as we have assumed that $P_1 \in  \TTW$. Thus our assumption,
  namely $a_1d$ passes through $e'$ must be false.





  \end{proof}
\fi







The above lemma can be used to discard bad paths from $\TTW$.
For each such discarded path, there exists at least one good path.
And by the above lemma, each such good path can be used to discard
at most one bad path. Thus the number of good paths in $\TTW$
is $\ge$ number of bad paths in $\TTW$.
We have already shown that the total number of good paths in $\TTW$
is $O(\sqrt{n\sigma})$.
Thus the total number of paths in $\TTW$ is also $O(\sqrt{n\sigma})$.


\section{Building the Data Structure}
\label{sec:data}

Let us first recognize a potential problem in using $\TON(\cdot)$.
Let $s_1t$ and $s_2t$ path  meet at vertex $x$ (See Figure \ref{fig:heavylight}).
Another path $s_3t$
 meets $s_2t$ path at $y$ where $y$ is closer to $t$.
$\TON(s_2x)$ is the shortest path from $s_1$ to $t$ avoiding $e$ and $\TON(xy)$ is the shortest path from
$x$ to $t$ avoiding $e \in yt $. This immediately leads to the following problem. Assume that
the query is $\textsc{Q}(s_2,t,e)$ and the preferred path avoiding $e$ is in $\TON$.
Then there are two  candidate paths that avoid
$e$: one that goes from $s_2$ to the intersection vertex $x$ and then take
path  $\TON(xy)$ and the other $\TON(s_2x)$.
Thus, we need to check these two paths and return the minimum of the two. One can
make a bigger example in which there are $\sigma$ segments between $s_2$ and $t$
and thus we have to check $O(\sigma)$ path before we can answer the query. The problem appears
because we don't know from which segment the shortest path avoiding $e$ started its detour.
If this information is not there, then it seems that we have to look at all the segments between
$s_2$ ans $t$.
\iflong
\begin{figure}[hpt!]
\centering

\begin{tikzpicture}[scale=1.6]
\begin{scope}
\coordinate (s2) at (0.4,2);
\coordinate (s1) at (-0.4,2);
\coordinate (s3) at (1.1,2);
\coordinate (x) at (0,1.4);
\coordinate (y) at (0,0.8);
\coordinate (t) at (0,0);
\coordinate (ts) at (0.15,0.4);
\coordinate (b1) at (0,.2);

\coordinate (a1) at (0,1.2);
\coordinate (a2) at (0.13,1.6);
\coordinate (v) at (0,.6);

\coordinate (c) at (-0.7,0.66);
\coordinate (b2) at (0,0.1);

\draw[thick](s1)--(x);
\draw[thick](s3)--(y);
\draw[thick](s2)--(x);
\draw[thick](x)--(t);
\node[above] at (s2){$s_2$};
\node[above] at (s1){$s_1$};
\node[above] at (s3){$s_3$};
\node[below] at (t){$t$};

\node[left] at (y){$y$};
\node[right] at (x){$x$};

\draw[blue,thick] (a1) to[out=180,in=180,distance=.8cm]
node[pos=0.4,left]
{\scriptsize  $\TON(xy)$}  (b1);

\draw[red,thick] (a2) to[out=10,in=10]
node[pos=0.5,right]
{\scriptsize  $\TON(s_2x)$}  (b2);

\node at (v){$\times$};
\node[right] at (v){$e$};

\end{scope}

\end{tikzpicture}
\caption{The shortest path from $s_2$ to $t$ avoiding $e$ can be $\TON(xy)$ or $\TON(s_2x)$.}
\label{fig:heavylight}

\end{figure}
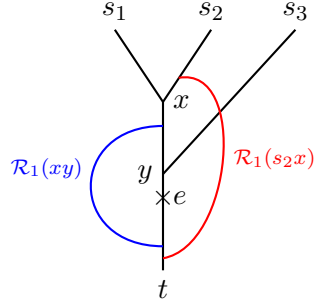
\fi
To end this dilemma, we use heavy light decomposition of \SBFS($t$) \cite{SleatorT83}.
For any segment $xy \in$ \SBFS($t$) (by our convention $y$ is closer to $t$), $x$ is a {\em heavy}
child of $y$ if the number of nodes in the subtree under $x$ is $\ge$ 1/2(number of
nodes in the subtree under $y$) else it is called a {\em light child} (or light {\em segment} in our case).
It follows
that each intersection vertex has exactly one heavy child and
each vertex is adjacent to atmost two heavy edges. A {\em heavy chain} is a concatenation of
heavy edges. A {\em heavy subpath} is a subpath of a heavy chain.
The following lemma notes a well known property of heavy-light decomposition.

\begin{lemma}
\label{lem:decomposition}
The path from a source $s$ to $t$ in \SBFS($t$) can be decomposed into $O(\log n)$ heavy subpaths
and light segments .
\end{lemma}

Given any source $s \in S$, by Lemma \ref{lem:decomposition},
the path from $t$ to $s$
may contain many heavy subpaths.
Let $C(pq)$ be a heavy
chain that starts at $p$
and ends at $q$ (where $q$ is closer to $t$ than $p$). A $ts$ path may follow
a heavy chain $C(pq)$ but may exit
this chain from a vertex midway, say at $r$. Let $(C(pq), r)$
be a tuple associated with $s$
such that the shortest path from $t$ to $s$ enters this
heavy chain via $q$ and
leaves this chain at $r$. We keep a list $\HE(s,t)$
which contains
all the tuples $(C(pq), r)$ sorted according to
the distance of heavy chain
from $t$ (that is distance $qt$). By Lemma \ref{lem:decomposition},
the size
of $\HE(s,t) = O(\log n)$. Similarly, we have one more
list to store the light segments.
$\LI(s,t)$ contains all the light segments on the $st$ path
again ordered according to their
distance from $t$ in \SBFS($t$). Again by Lemma  \ref{lem:decomposition},
the size
of $\LI(s,t) = O(\log n)$. Note that the size of these additional
two data-structures
is $\sum_{s \in S} O(\log n) = \tilde O(\sigma) = \tilde O(\sqrt{n\sigma})$.

Our main problem was that we have to find the minimum $\TON(\cdot)$ of $O(\sigma)$ segments if 
there is a path of length $\sigma$ between $s$ and $t$. The trick we use here is that
finding minimum on any heavy subpath takes $\tilde O(1)$ time. Since
there are $O(\log n)$ heavy subpaths, the total time taken to find the minimum
on heavy subpaths in $\tilde O(1)$. Also, since the number of light segments is also $O(\log n)$
finding the minimum among these also takes $\tilde O(1)$ time.

We now describe our intuition in detail.  Let $xy$ be a segment in a heavy chain
$C(pq)$. We want to represent $\TON(xy)$ in a
balanced binary search tree $\BST(C)$. To this end, we will add a node with
the tuple $(x.depth,| px \conc \TON(xy) |, |px|)$  in $\BST(C)$.
The first element in this tuple is the depth of $x$
in $\BFS(t)$ --  it also acts as the key in this binary search tree. The second element
is the path $ \TON(xy)$  concatenated with $px$. This concatenation is done so that all  paths
in $\BST(C)$ start from $p$ and comparing two paths in $\BST(C)$ is
possible. The third element will be used to get the path length $\TON(xy)$ (by subtracting it from
the second element) when
need arises. Now we can augment
this tree so that the following range minimum query can be answered in $\tilde O(1)$ time:
$\RMQ(C(pq),[a,b])$ : Find minimum of $\{|px \conc \TON(xy)|\ |\ xy \ \text{is a segment in heavy chain $C(pq)$ and }\ x.depth \ge a.depth \ \text{and}\  x.depth \le b.depth \}$.
The size of $\cup_{C \in \HE(s,t)}\BST(C)$ is $O(\sigma) = O(\sqrt{n\sigma})$ as there are
at most $O(\sigma)$ segments in \SBFS($t$).
\iflong
\begin{algorithm}
\SetKwRepeat{Do}{do}{while}%
  \caption{Finding the shortest replacement path  (in $\TON$) from $s$ to $t$ avoiding $e(u,v)$ }
  \label{fig:findR1R2}
Let $x$ be  the first intersection point on $us$ path\;
$min \leftarrow \infty$ \;

\Do{$x$ is not equal to $s$ }
{
    \If{ $x$ lies in a heavy chain}
    {
        Using $\HE(s,t)$, find $(C(pq),r)$, that is $r$ is the vertex from which $us$
        path leaves the chain $C$\;
        $min\leftarrow \min\{ min, \RMQ(C, [x,r]) - |pr| + |sr| \}$\;
        $x \leftarrow r$\;
    }
    \ElseIf{$x$ is an endpoint of a light segment}
    {
        Let $x'x$ be the  light segment ending at $x$ in \SBFS($t$) \tcp*{Can be found out via $\LI(s,t)$ in $\tilde O(1)$ time.}

        $min \leftarrow \min\{ min, |\TON(x'x)| + |sx'| \}$\;
        $x \leftarrow x'$\;
    }
}
\end{algorithm}
\fi

Given any edge $e(u,v)$ on $st$ path, we can now find the shortest path in $\TON$ from
$s$ to $t$ avoiding $e$
\iflong
(see Algorithm \ref{fig:findR1R2}).
\fi We first find the first intersection vertex on the
$us$ path from $u$. Let this vertex be $x$. We will see that finding
$x$ is also not a trivial problem --  we will say more about this problem later.
Now, we will go over all  possible replacement paths from $u$ to $s$.
Thus, we  search if there exists any heavy chain in
$\HE(s,t)$ that contains $x$. To this end, we first check if $x$ lies in some light segment (this can be checked in $\tilde O(1)$ time). If not, then $x$ lies in some heavy chain. We now search each heavy chain in $\HE(s,t)$ to find a node $x'$ with the smallest depth such that  $x'.depth > x.depth$.
Let this node be $x'$. Thus we have found the segment $x'x$ where $x$ is closer to $t$
than $x'$.
We can easily calculate $x.depth$ as
$|st| - |sx|$ or $B_0(s,t) - B_0(s,x)$. Since there are $\tilde O(1)$ heavy
chain in $\HE(s,t)$, the time taken to find if $x'x$ exists in some heavy chain is $\tilde O(1)$.

Assume that we found out that
$x'x \in C(pq)$,  and $ts$ path leaves the chain $C$ at $r$, then we want to
find the shortest replacement path from $r$ to $t$ avoiding $e$. This can be found out via the range minimum query $\RMQ(C(p,q),[x,r])$.
However, note that each replacement path in $C$ starts from $p$. So, we need to
remove $|pr|$ from the replacement path length returned by $\RMQ$ query.
The length $pr$ can be found out in the node $r \in \BST(C)$. Finally, we add $|sr|$ to get the path from $s$ to $t$.

Similarly, we can process a light segment in $O(1)$ time (please refer to Algorithm \iflong\ref{fig:findR1R2}\fi~).
Thus, the time taken by  Algorithm  \iflong\ref{fig:findR1R2}\fi~ is
$\tilde O(1)$ as the while loop runs at most $O(\log n)$ times  and each step in
the while loop runs in $\tilde O(1)$ time.

\subsection{Answering queries in $\tilde O(1)$ time}
Given a query ${\sc Q}(s,t,e(u,v))$, we process it as follows
(assuming that $e$ lies on $st_s$ path (that is the {\em far case)} and $v$ is closer to $t$ than $u$)\iflong\\\fi

\begin{enumerate}[leftmargin=*,noitemsep,nolistsep]
\item Find the first intersection vertex  on $us$ path.

\iflong
  \noindent As we have mentioned before, this is not a trivial problem.
  Let the first intersection vertex from $u$ to $s$ in $\BST(t)$ be
  denoted by  $\INT_{s}(u,t)$.   We will first show that $\INT_{s}(u,t)$ is independent of $s$.
  \begin{lemma}
  \label{lem:intlemma}
  $\INT_s(u,t) = \INT_{s'}(u,t)$ for $s,s' \in S$ for all $u \in \BFS(t)$.
  \end{lemma}

  \iflong
  \begin{proof}
    We will prove this by induction on the nodes of $\BFS(t)$ from leaf to root $t$.
    The base case is a leaf in $\BFS(t)$, that is a source vertex, which by definition,
    itself is an intersection vertex. For any node $u$, if $u$ is an intersection vertex,
    $\INT_s(u,t) = \INT_{s'}(u,t) = u$. Else, $u$ is a node of degree 2 in $\BFS(t)$.
    Assume that $u'$ is the child of $u$ in $\BFS(t)$. So, $\INT_s(u,t) = \INT_s(u',t)$ and $\INT_{s'}(u,t) = \INT_{s'}(u',t)$. But by induction hypothesis,  $\INT_s(u',t)=\INT_{s'}(u',t)$.

    \end{proof}
  \fi

  \noindent We will drop the subscript $s$ from the definition of $\INT(\cdot,\cdot)$ as we now know that it is independent of $s$. We use the
  above lemma to construct  two data structures that will help us in finding $\INT(u,t)$.

  \begin{itemize}
  \item $I_1(t)$: For any $u \in V$, if $\INT(u,t)$ is within a distance of $c\sqrt{n/\sigma} \log n$
  (for some constant $c$) from $u$, then we store the tuple
  $(u,\INT(u,t)$) in the balanced binary search tree $I_1(t)$. For any  intersection
  vertex $x \in $ \SBFS$(t)$, we  store at most $\tilde O(\sqrt{n/\sigma})$ tuples in $I_1(t)$.
  For a fixed $t$, the total
  space taken by $I_1(t)$ is $\tilde O(\sigma .\sqrt{(n/\sigma)}
  )= \tilde O(\sqrt{n \sigma})$ (as there are $O(\sigma)$ intersection vertices in \SBFS($t)$.

  \item $I_2(t)$: If $u$ is not present in $I_1(t)$, then $\INT(u,t)$ is at a distance
  $\ge c\sqrt{n /\sigma}\log n$ from $u$. We now
  use a different strategy to find $\INT(u,t)$. We first find $u_s \leftarrow B_1(s,u)$,
  that is the vertex in $\TT$
  closest to  $u$ in $su$ path. With a
high probability,
$u_s$ is closer to $u$ than $\INT(u,t)$ and all the vertices from $u$ to $u_s$ have degree  exactly 2 in $\BFS(t)$.
  Thus, $\INT(u_s,t) = \INT(u,t)$ -- we will now use this property (a similar property was also used in the proof of Lemma \ref{lem:intlemma}).

  \noindent
  For each $x \in \TT$
  such that $x$ is not an intersection vertex in $\BFS(t)$, we store the tuple
  $(x,\INT(x,t))$ in another balanced binary search tree $I_2(t)$. For a fixed vertex $t$, the
  size of this data-structure is $\tilde O(\sqrt{n\sigma})$ space
  as there is only one intersection vertex for each vertex
  in $\TT$ and $|\TT| = \tilde O(\sqrt{n \sigma})$

\end{itemize}
  \noindent If $\INT(u,t)$ is indeed at a distance $\le
c \sqrt{n / \sigma} \log n$,
  then we can use $I_1(t)$  to find it in $\tilde O(1)$
time, else we use $I_2(t)$ to find
  $\INT(u_s,t)$ in $\tilde O(1)$ time.\\

\fi
\item Find  the replacement path avoiding $u$  if it lies in $\TON$.
\label{fcase1}

\noindent To this end, we use our Algorithm \iflong\ref{fig:findR1R2}\fi.
The first non-trivial part of this algorithm, that is, finding the first
intersection vertex on the $us$ path has already been tackled in the point above.
So we can find such a replacement path (if it exists) in
$\tilde O(1)$ time  and $\tilde O(\sqrt{n \sigma})$ space.\iflong\\\fi

\item Find the replacement path avoiding $e(u,v)$ if it lies in $\TTW$.
\label{fcase2}

\noindent  This part is similar to our data-structure in single source case.
Let $x \leftarrow \INT(u,t)$.  Using $\HE(s,t)$ and $\LI(s,t)$, in $\tilde O(1)$ time,
we can find the segment $xy \in $ \SBFS($t$) such that $y$ is closer to $t$ than $x$.
In this case, we want to check if there exists any
replacement path that starts in  the  same segment in which
$e$ resides. This  replacement path first takes $sx$ path and then takes
 the  detour strictly inside  the  segment $xy$. All such paths are stored in $\TTW(xy)$
with the contiguous range of edges that they avoid on $xy$.
We now just need to check if $u$ and $v$ lie in the range of some replacement path.
To this end, we find $u.depth \leftarrow |st| -|su| = B_0(s,t) - B_0(s,u)$
and $v.depth \leftarrow |st| -|sv| = B_0(s,t) - B_0(s,v)$.
Now we check if $u.depth$ and $v.depth$ lie in contiguous range of
some replacement path  in $\TTW(xy)$. If yes, then we return the length of that
path concatenated with $sx$.  Note that we have already stored $|sx|$ in $B_0(s,x)$.
The time taken in this case is dominated by searching $u$ and $v$ in $\TTW(xy)$,
that is $\tilde O(1)$.\iflong\\\fi
\end{enumerate}

\noindent Thus, the total query time of our algorithm is $\tilde O(1)$, and we can return the
minimum of replacement paths found in Step \ref{fcase1} and \ref{fcase2} as our final answer.
The reader can check that the space taken by our algorithm for a
vertex $t$ is $\tilde O(\sqrt{n\sigma})$. Thus the total space taken by our algorithm is
$\tilde O( \sigma^{1/2} n^{3/2}) $. Thus we have proved the main result, that is Theorem \ref{thm:maintheorem} of our paper.


\bibliographystyle{plain}
\bibliography{sample}

\end{document}